\documentclass[journal,twoside,web]{ieeecolor}
\usepackage{generic}
 \usepackage{cite}
 \usepackage{amsmath,amssymb,amsfonts}
\usepackage{algorithmic}
\usepackage{graphicx}
\usepackage{algorithm,algorithmic}
\usepackage{hyperref}
\hypersetup{hidelinks=true}
\usepackage{textcomp}
\usepackage{color}
\usepackage{float}
\usepackage{subfigure}
\usepackage{caption}
\usepackage{cleveref}

\DeclareCaptionFormat{custom}{#1#2#3} 
\newtheorem{theorem}{Theorem}
\newtheorem{assumption}{Assumption}
\newtheorem{definition}{Definition}
\newtheorem{lemma}{Lemma}

\newtheorem{remark}{Remark}
\newtheorem{corollary}{Corollary}
\newtheorem{example}{Example}

\allowdisplaybreaks[4]
\def\BibTeX{{\rm B\kern-.05em{\sc i\kern-.025em b}\kern-.08em
    T\kern-.1667em\lower.7ex\hbox{E}\kern-.125emX}}
\begin{document}
\title{Learning Control for LQR with Unknown Packet Loss Rate Using Finite Channel Samples}
\author{Zhenning Zhang, Liang Xu, \IEEEmembership{Member, IEEE}, Yilin Mo, \IEEEmembership{Member, IEEE}, and Xiaofan Wang, \IEEEmembership{Senior Member, IEEE}
\thanks{Zhenning Zhang and Xiaofan Wang are with the School of Mechatronic Engineering and Automation, Shanghai University, Shanghai, China. Email: mathzzn@163.com, xfwang@shu.edu.cn}
\thanks{Liang Xu is with the School of Future Technology, Shanghai University, Shanghai, China. Email: liang-xu@shu.edu.cn (corresponding author: Liang Xu)}
\thanks{Yilin Mo is with the Department of Automation and BNRist, Tsinghua University, Beijing, China. Email: ylmo@tsinghua.edu.cn}
\thanks{The work was supported in part by the National Natural Science Foundation of China under Grant 62273223, 62373239, 62333011, and the Project of Science and Technology Commission of Shanghai Municipality under Grant 22JC1401401.}
}

\maketitle

\begin{abstract}
This paper studies the linear quadratic regulator (LQR) problem over an unknown Bernoulli packet loss channel. 
The unknown loss rate is estimated using finite channel samples and a certainty-equivalence (CE) optimal controller is then designed by treating the estimate as the true rate. 
The stabilizing capability and sub-optimality of the CE controller critically depend on the estimation error of loss rate.
For discrete-time linear systems, we provide a stability threshold for the estimation error to ensure closed-loop stability, and analytically quantify the sub-optimality in terms of the estimation error and the difference in modified Riccati equations.
Next, we derive the upper bound on sample complexity for the CE controller to be stabilizing.
Tailored results with less conservatism are delivered for scalar systems and n-dimensional systems with invertible input matrix. 
Moreover,  we establish a sufficient condition, independent of the unknown loss rate, to verify whether the CE controller is stabilizing in a probabilistic sense. 
Finally, numerical examples are used to validate our results.
\end{abstract}

\begin{IEEEkeywords}
LQR, Bernoulli packet drop, stability threshold, sample complexity, performance analysis.
\end{IEEEkeywords}

\section{Introduction}
Networked control systems (NCSs) are widely adopted for their cost-efficiency and flexibility, yet they also introduce new challenges \cite{NCSs survey 1,NCSs survey 2,NCSs survey3,NCSs survey5,NCSs6,NCSs survey4}.
Packet drop, as a common problem in NCSs, disrupts the connection loop between different devices \cite{packet drop 1,cyber1,bernoulli1,kalman2005,06optimal,foundation,Zhang08Modelling,Xiong07Stabilization}. The seminal work \cite{kalman2005} firstly considers Kalman filtering with packet loss which follows a Bernoulli distribution. Subsequently, control problems over lossy channels were studied, e.g., \cite{06optimal,foundation}. The work \cite{06optimal} investigates the linear quadratic regulator (LQR) over Bernoulli packet loss channels, where packets in sensor-controller (SC) and controller-actuator (CA) channels may be lost. The work \cite{foundation} extends LQR with packet drop to the Linear Quadratic Gaussian (LQG) problem and proves that the separation principle still holds under the TCP protocol. 
According to \cite{06optimal,foundation}, we know that the optimal controller over lossy channels depends on the channel information, i.e. the packet drop probability. Additionally, the works \cite{Zhang08Modelling,Xiong07Stabilization,packet drop 1} present suboptimal controllers for NCSs with packet loss compensation mechanisms, which also rely on channel information.

However, obtaining channel information may involve high costs or even be impossible \cite{channel information,unknown channel 3}. Many studies have begun to exploit NCSs over unknown channels \cite{unknown channel 1,unknown channel 2,unknown channel 3,unknown channel 4,unknown channel 5,unknown channel 6,unknown channel 7}. 
The work \cite{unknown channel 4} considers control problems over multiple channels with unknown Bernoulli packet drop and employs Thompson sampling to learn loss probabilities, without restrictions on sample size. Yet in practice, only a finite set of data is available. Novel perspectives based on finite-sample learning and non-asymptotic analysis provide a more detailed description of the learning difficulty \cite{finite sample,finite sample2}. 
The works \cite{unknown channel 5,unknown channel 6} consider a linear system with unknown Bernoulli packet drop, which is stable if the drop probability is sufficiently low and unstable if it is not. They use finite samples to estimate the unknown probability and then answer the system is stable or not on a confidence level. The key of these works is Hoeffding's inequality which establishes a bridge between sample size and estimation error. But the controller design problems have not been considered. Furthermore, \cite{unknown channel 7} considers LQR over an unknown Bernoulli lossy channel. They utilize finite samples to construct a confidence interval for the unknown true loss probability and then design a certainty-equivalence (CE) optimal controller based on the worst-case (i.e., the maximum) loss probability within the confidence interval. The goal of [22] is to provide upper and lower bounds on controller performance that can be computed in practice, independent of the unknown true loss probability. In fact, it also indicates that the worst-case loss probability lies within the stabilizing set of estimated loss probabilities ensuring closed-loop stabilization via the CE controller, as the infinite-horizon cost function of the ``worst-case" CE controller is bounded.  

Building upon \cite{unknown channel 7}, a further question arises: how to determine the entire stabilizing set of estimated loss probabilities that ensure closed-loop stability? In other words, is there a threshold on the estimation error that guarantees stabilization via the CE controller? Additionally, what is the sample complexity for the CE controller to be stabilizing? Furthermore, the estimation error of loss probability may create an optimality gap between the CE and the optimal controllers. While \cite{unknown channel 7} provides practical performance bounds tailored to the ``worst-case" CE controller, we are also interested in how to theoretically quantify the sub-optimality of the CE controller synthesized from any stabilizing estimated probability in terms of the estimation error.


This paper also consider LQR problem with unknown Bernoulli packet loss. Signals in the CA channel may be lost, resulting in a system input of zero. Finite channel samples are used to estimate loss probability, and the CE controller is then designed directly based on the estimated probability. The main contributions are listed below.

\begin{enumerate}
    \item First, a necessary and sufficient condition is provided for ensuring stabilization via the CE controller in scalar case. Then it is extended to $n$-dimensional systems as a sufficient condition. Based on the stabilizing conditions, we establish the existence of a stability threshold for the estimation error of loss probability, and explicitly express its lower bound for $n$-dimensional systems, below which the CE controller is stabilizing. Additionally, tailored results with less conservatism are provided for scalar systems and $n$-dimensional systems with an invertible input matrix.
    \item We derive the sample complexity for the CE controller to be stabilizing, and explicitly express its upper bound. Tailored results are also provided for scalar systems and $n$-dimensional systems with an invertible input matrix. Moreover, a practical sufficient condition is established to determine whether the CE controller is stabilizing, independent of the unknown true loss probability.
    \item We quantify the sub-optimality of the CE controller incurred in terms of estimation error of loss probability. The optimality gap is analytically expressed as a linear combination of the estimation error and the difference between the solutions of modified Riccati equations corresponding to the estimated and true probabilities.
\end{enumerate}

The remainder of this paper is organized as follows. Section \ref{section2} introduces problem formulations; Section \ref{section3} analyzes estimation error threshold; Section \ref{section4} presents sample complexity; Section \ref{section5} provides performance analysis; Section \ref{section6} offers numerical simulations; Section \ref{section7} concludes the paper.

\textit{Notation:} $\mathbb{R}^{n\times m}$ is the $n\times m$ dimensional real matrix set; $\mathbb{R}_+$ is the set of non-negative real numbers; $\rho(M)$ is the spectral radium of a matrix $M\in\mathbb{R}^{n\times n}$; $\lambda_{\min}(M)$ is the minimum eigenvalue of a real symmetric matrix $M$; $\lambda_{\max}(M)$ is the maximum eigenvalue of a real symmetric matrix $M$; $M^{\top}$ is the transpose of a matrix $M$; $\mathcal{I}$ is the identity matrix with appropriate dimensions; $\mathcal{P}(\cdot)$ is probability measure; $\log(\cdot)$ is the logarithm of a positive real number in base $e$.

\section{Preliminaries}\label{section2}
\subsection{Problem Formulation}
Considering the linear time-invariant system
\begin{equation}\label{1}
    x_{t+1}=Ax_t+\lambda_tBu_t,~t=0,1,2,\dots
\end{equation}
where $x_t\in\mathbb{R}^n$ is the system state, $x_0$ is the initial value with finite expectation and variance, $u_t\in\mathbb{R}^{m}$ is the control input, $A\in\mathbb{R}^{n\times n}$ is the state matrix, $B\in\mathbb{R}^{n\times m}$ is the input matrix; $\{\lambda_t\}_{t=0}^{+\infty}$ is a sequence of i.i.d. Bernoulli random variables satisfying the following distribution 
\begin{equation}\label{lambda distribution}
\lambda_t=\left\{\begin{array}{cc}
     0~~~w.p.&q  \\
     1~~~w.p.&1-q
\end{array}\right.,    
\end{equation}
where $q\in(0,1)$. $\{\lambda_t\}_{t=0}^{+\infty}$ models the packet losses that may occur in the communication channel \cite{06optimal}. At each instant, the packet of control signal may be lost with a probability of $q$ (i.e., loss rate), or successfully transmitted to the system with a probability of $1-q$. If the control packet is lost, the system input becomes zero at the current time. The closed-loop system with packet drop is shown in Fig. \ref{closed loop}.
\begin{figure}[h]
    \centering
    \includegraphics[width=0.8\linewidth]{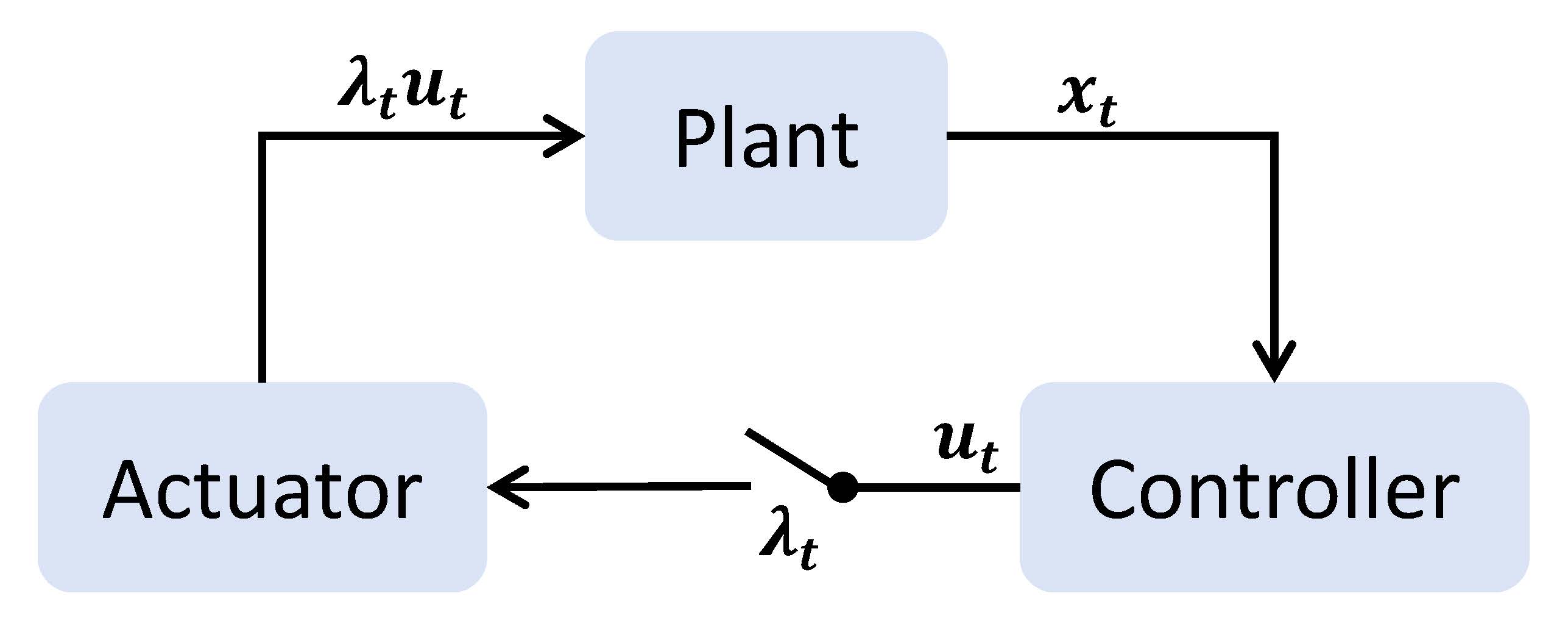}
    \caption{The closed-loop system with packet drop.}
    \label{closed loop}
\end{figure}

The cost function is 
\begin{equation}\label{2}
    J(x_0,\textbf{u})=\mathbb{E}\Bigg\{\sum_{t=0}^{\infty}x_t^{\top}Qx_t+\lambda_t u_t^{\top}Ru_t\Bigg\},
\end{equation}
where $\textbf{u}=\{u_0,u_1,u_2,\dots\}$ is the control input sequence, $Q\in\mathbb{R}^{n\times n},R\in\mathbb{R}^{m\times m}$ are positive definite matrices, and the expectation is taken over $\{\lambda_t\}_{t=0}^{+\infty}$.
The objective of the LQR problem is to design a control sequence $\textbf{u}^*$ that mean-square stabilizes the system \eqref{1} and minimizes the cost function $J(x_0,\textbf{u})$ \cite{coppens1}. 
The mean-square stability in this paper is defined as follows.
\begin{definition}
	The System \eqref{1} is mean-square stable if $\lim_{t\to\infty}\mathbb{E}\{x_t^{\top}x_t\}=0$.
\end{definition}

We denote the minimum cost function as $J^*(x_0)$ and the optimal control as $\textbf{u}^*=\{u_0^*,u_1^*,u_2^*,\dots\}$, i.e., $J^*(x_0)=\min_{\textbf{u}}J(x_0,\textbf{u})=J(x_0,\textbf{u}^*)$.
For the infinite horizon LQR problem with Bernoulli packet drop \cite{06optimal,foundation}, the optimal controller depends on prior knowledge of the true packet drop probability $q$, which is 
$u^*_t=Kx_t$, where $K=-(R+B^{\top}PB)^{-1}B^{\top}PA$, and $P$ is a $n$-dimensional positive definite matrix satisfying the modified Riccati equation
\begin{equation}\label{RE1}
    P=Q+A^{\top}PA-(1-q)A^{\top}PB(R+B^{\top}PB)^{-1}B^{\top}PA,
\end{equation}
with the optimal cost $J^*(x_0)=\mathbb{E}[x_0^\top Px_0]$.
To ensure the existence of the positive definite solution to \eqref{RE1}, it is necessary to assume that $(A,B)$ is stabilizable. Moreover, the probability of packet drop cannot exceed a certain critical value \cite{foundation}, which is stated in the following assumption.

\begin{assumption}\rm\label{assump1}
    The packet drop probability $q$ satisfies $q<q_c$,
where $q_c$ satisfies
\begin{equation}\label{range of q_c}\frac{1}{\prod_i|\lambda_i^u(A)|^2}\le q_c\le\frac{1}{\max_i|\lambda_i^u(A)|^2},\end{equation}
and $\lambda_i^u(A)$ denote the unstable eigenvalues of $A$.
\end{assumption}


In Assumption \ref{assump1}, $q_c$ is an upper bound for $q$ to ensure the existence of a positive definite solution $P$ to the equation \eqref{RE1}, which depends on the system matrix $A$ but cannot be expressed in explicit form. It lies within the certain range indicated in \eqref{range of q_c}. Specially, $q_c=1/\max_i|\lambda^u_i(A)|^2$ when $B$ is an invertible matrix; $q_c=1/\prod_i|\lambda_i^u(A)|^2$ when $B$ has rank one \cite{foundation}.


The dependence of $P$ on $q$ indicates that $q$ plays a critical role in obtaining the optimal controller $\textbf{u}^*$. In practice, $q$ may be unknown, and we can only use channel samples to estimate $q$. Denote $\hat{q}$ as the estimated value of $q$. Given a sample sequence $\{\lambda_i,i=1,\dots,N_q\}$ of i.i.d. Bernoulli random variables, $\hat{q}$ is computed as $\hat{q}=\frac{1}{N_q}\sum_{i=1}^{N_q}(1-\lambda_i)$ \cite{unknown channel 7,boss1}. Then, we design the CE optimal controller $\hat{\textbf{u}}=\{\hat{u}_0,\hat{u}_1,\hat{u}_2,\dots\}$ based on $\hat{q}$ (i.e., treat $\hat{q}$ as the true probability $q$) and $\hat{u}_t$ is
\begin{equation}\label{hat u}
\hat{u}_t=\hat{K}x_t,~t=0,1,2,\dots
\end{equation}where
\begin{equation}\label{hat K}
    \hat{K}=-(R+B^{\top}\hat{P}B)^{-1}B^{\top}\hat{P}A,
\end{equation}
and $\hat{P}\in\mathbb{R}^{n\times n}$ is the positive definite solution of the following modified Riccati equation
\begin{equation}\label{RE2}
    \hat{P}=Q+A^{\top}\hat{P}A-(1-\hat{q})A^{\top}\hat{P}B(R+B^{\top}\hat{P}B)^{-1}B^{\top}\hat{P}A.
\end{equation}
The sample size $N_q$ should at least ensure $\hat{q}<q_c$, such that \eqref{RE2} has a positive definite solution $\hat{P}$ used to design $\hat{K}$. 
Although a suitable $N_q$ theoretically guarantees $\hat{q}<q_c$ only in a probabilistic sense (since $\hat{q}$ is random), many existing studies can be used in practice to test whether this condition holds. For some special linear systems, such as those with an input matrix of rank $1$ or full rank, the exact value of $q_c$ can be obtained (as discussed below \Cref{assump1}) and used to directly verify $\hat{q}<q_c$. For general linear systems, a bisection method based on LMI formulation of \eqref{RE2} can be employed to search for $q_c$ (as described in \cite{unknown channel 7}), and subsequently check whether $\hat{q}<q_c$ holds. If it does not, an increase in sample size is required. Therefore, this paper does not further discuss the conditions under which $\hat{q}<q_c$ holds. Instead, we treats $\hat{q}<q_c$ as a premise to investigate the main issues of concern, which will be introduced below.
We will refer to $\hat{K}$ and $\hat{u}_t$ interchangeably as the CE optimal controller.


Since samples are finite, there inevitably exists an estimation error between $q$ and $\hat{q}$. The next subsection will demonstrate that when the error is too large, $\hat{K}$ does not mean-square stabilize the system.

\subsection{Scalar Stability Analysis and Motivating Example}
In this subsection, we consider the scalar case and provide the necessary and sufficient condition for \eqref{hat K} to mean-square stabilize the system \eqref{1} when $n=m=1$. It is straightforward to verify that the necessary and sufficient condition is violated when the estimation error $q-\hat{q}$ is positive and large.

\begin{theorem}\label{thm n and s scalar}
    When $n=m=1$, with the Assumption \ref{assump1} and $\hat{q}<q_c$ holding, \eqref{hat K} mean-square stabilizes the system \eqref{1} if and only if 
   \begin{equation}\label{scalar condition}
    Q+(1-q)R\hat{K}^2+(\hat{q}-q)(R+B^2\hat{P})^{-1}A^2B^2\hat{P}^2>0.
\end{equation}
Moreover, if $\hat{q}\ge q$, \eqref{scalar condition} holds.
\end{theorem}
\begin{proof}
    When $n=m=1$, \eqref{hat K} and \eqref{RE2} can be rewritten as 
    \begin{align}
    &\hat{K}=-(R+B^2\hat{P})^{-1}AB\hat{P}.\label{scalar K}\\
    \label{scalar P}
    &\hat{P}=Q+A^2\hat{P}-(1-\hat{q})(R+B^2\hat{P})^{-1}A^2B^2\hat{P}^2,    
\end{align}
If the scalar system \eqref{1} is mean-square stabilized via \eqref{hat K}, we have $\mathbb{E}\{x_{t}^2\}\to 0$ as $t\to \infty$, which is equivalent to $\mathbb{E}\{\hat{P}x_t^2\}\to0$ as $t\to\infty$, where $\hat{P}$ is a positive real number satisfying \eqref{scalar P}.
Since
$\mathbb{E}\{\hat{P}x_{t+1}^2\}=[A^2+(1-q)(B^2\hat{K}^2+2AB\hat{K})]^{t+1}\mathbb{E}\{\hat{P}x^2_0\}$ and $\lim_{t\to\infty}\mathbb{E}\{\hat{P}x_t^2\}=0$, we have
$A^2+(1-q)(B^2\hat{K}^2+2AB\hat{K})<1$
and $\mathbb{E}\{\hat{P}x_t^2\}$ is monotonically decreasing, which implies that for any $t\ge0$, $\mathbb{E}\{\hat{P}x_{t+1}^2\}<\mathbb{E}\{\hat{P}x_t^2\}$.
On the other hand, we have
\begin{align}
        \mathbb{E}\{\hat{P}x_{t+1}^2\}&=\mathbb{E}\{[A^2+(1-q)(B^2\hat{K}^2+2AB\hat{K})]\hat{P}x_t^2\}\notag\\
        &\overset{(a)}{=}\mathbb{E}\{\hat{P}x_t^2\}-[Q+(1-q)R\hat{K}^2+(\hat{q}-q)\times\notag\\
        &\quad (R+B^2\hat{P})^{-1}A^2B^2\hat{P}^2]\mathbb{E}\{x_t^2\},\label{scalar result}
\end{align}
where $(a)$ follows from \eqref{scalar K},\eqref{scalar P}. Therefore, $\mathbb{E}\{\hat{P}x_{t+1}^2\}<\mathbb{E}\{\hat{P}x_t^2\}$ is equivalent to \eqref{scalar condition}. Moreover, if $\hat{q}\ge q$, \eqref{scalar condition} always holds because all terms are positive.
\end{proof}




 Theorem \ref{thm n and s scalar} implies that $\hat{K}$ mean-square stabilizes the scalar system \eqref{1} when $q-\hat{q}\le0$. In fact, this conclusion also holds for $n$-dimensional systems, which will be proven later. However, when $q-\hat{q}>0$, it is easy to verify that \eqref{scalar condition} may be violated. For example, $q=0.4$, $\hat{q}=0$, $A=1.5$ and $B=Q=R=1$. Therefore, when the estimated error $q-\hat{q}$ is positive and large, $\hat{K}$ is not stabilizing.
Inspired by the above example, this paper focuses on three problems: for $\hat{K}$ to be stabilizing, what is the upper bound of the estimation error $q-\hat{q}$, what is the minimum sample size required, and how to quantify the sub-optimality of $\hat{K}$. The stability analysis up to a certain estimation error $q-\hat{q}$ (discussed in the next section) is important because overestimating the loss probability is conservative, even though it ensures stability.

\section{Stability analysis}\label{section3}
We will first analyze the stability threshold of $q-\hat{q}$ for $n$-dimensional systems, below which $\hat{K}$ is stabilizing. Then, tailored results with less conservatism are given for some specific cases.
\subsection{General Results}
First, we provide the following theorem which guarantees \eqref{hat K} mean-square stabilizes the system \eqref{1}.
\begin{theorem}\rm \label{Pro1}
    With Assumption \ref{assump1} and $\hat{q}<q_c$ holding, \eqref{hat K} mean-square stabilizes the system \eqref{1}, if $q,\hat{q}$ satisfy
    \begin{equation}\label{th1_19}
        Q\!+\!(\!1\!-q)\hat{K}^{\top}\!R\hat{K}\!-\!(\!q\!-\!\hat{q})\!A^{\top}\!\hat{P}\!B(\!R\!+\!B^{\top}\!\hat{P}\!B)^{-1}\!B^{\top}\!\hat{P}\!A\!>\!0
    \end{equation}
    where $\hat{P}$ is the positive definite solution of \eqref{RE2}. Moreover, if $\hat{q}{\ge}q$, \eqref{th1_19} holds.
\end{theorem}
\begin{proof}
    Introducing a Lyapunov function $\hat{V}(x_t)=\mathbb{E}[x_t^{\top}\hat{P}x_t]$, and we can derive that 
    \begin{align}
        \hat{V}(x_{t+1})=&\mathbb{E}\Big\{(Ax_t+\lambda_tB\hat{u}_t)^{\top}\hat{P}(Ax_t+\lambda_tB\hat{u}_t)\Big\}\notag\\
        \overset{(a)}{=}&\mathbb{E}\Big\{x_t^{\top}\big[A^{\top}\hat{P}A-(1-q)\hat{K}^{\top}R\hat{K}-(1-q)\times\notag\\
        &A^{\top}\hat{P}B(R+B^{\top}\hat{P}B)^{-1}B^{\top}\hat{P}A\big]x_t\Big\}\notag\\
        \overset{(b)}{=}&\mathbb{E}\Big\{x_t^{\top}\big[\hat{P}-Q-(1-q)\hat{K}^{\top}R\hat{K}+(q-\hat{q})\times\notag\\
         &A^{\top}\hat{P}B(R+B^{\top}\hat{P}B)^{-1}B^{\top}\hat{P}A\big]x_t\Big\}\notag\\
         \overset{(c)}{<}& \mathbb{E}\{x_t^{\top}\hat{P}x_t\}=\hat{V}(x_t).
    \end{align}
    where $(a)$, $(b)$, $(c)$ follow from \eqref{hat u}, \eqref{RE2}, \eqref{th1_19}, respectively.

Based on Lyapunov theory, we derive $\lim_{t\to\infty} \mathbb{E}\{x^{\top}_tx_t\}=0$, thus proving the mean-square stability of the system \eqref{1}. Moreover, since $Q>0$, $(1-q)\hat{K}^{\top}R\hat{K}>0$, $A^{\top}\hat{P}B(R+B^{\top}\hat{P}B)^{-1}B^{\top}\hat{P}A\ge0$, \eqref{th1_19} always holds if $\hat{q}{\ge}q$.
\end{proof}

The condition \eqref{th1_19} is a generalization of \eqref{scalar condition} from scalar systems to $n$-dimensional systems. \eqref{scalar condition} is a necessary and sufficient condition for $\hat{K}$ to stabilize the system in the scalar case. However, \eqref{th1_19} is only a sufficient condition for $n$-dimensional systems, as $\mathbb{E}\{x^{\top}_{t+1}x_{t+1}\}<\mathbb{E}\{x^{\top}_{t}x_{t}\}$ is not equivalent to $\mathbb{E}\{x^{\top}_{t+1}\hat{P}x_{t+1}\}<\mathbb{E}\{x^{\top}_{t}\hat{P}x_{t}\}$ for a given $\hat{P}$.


A larger $q$ represents a more hostile communication environment. Hence, the second conclusion of Theorem \ref{Pro1} indicates that conservative controllers synthesized with overestimated loss rates (i.e., $\hat{q}>q$) inherently guarantee closed-loop stability, albeit at the expense of increased sub-optimality (validated in \Cref{section5}).
Then, we focus on the stability threshold for the estimation error when $\hat{q}<q$.
For the sake of clarity, we define the \textit{Stability Threshold} (ST) as 
\begin{equation}\label{MST definition}
    \bar{\delta}=\max\Big\{\delta\in\mathbb{R}_+\Big|\mathcal{C}(q,\hat{q})>0,~ \forall (q,\hat{q})\in\Omega_{\delta}\Big\}.
\end{equation}
where $\mathcal{C}(q,\hat{q})=Q+(1-q)\hat{K}^{\top}R\hat{K}-(q-\hat{q})A^{\top}\hat{P}B(R+B^{\top}\hat{P}B)^{-1}B^{\top}\hat{P}A$, i.e., the left-hand side of \eqref{th1_19}, and $\Omega_{\delta}=\big\{(q,\hat{q})\in\mathbb{R}\times\mathbb{R}\big|q-\hat{q}<\delta,~0\le q,\hat{q}<q_c\big\}$.

From the definition \eqref{MST definition}, we know that when $q-\hat{q}<\bar{\delta}$, $\hat{K}$ mean-square stabilizes the system according to Theorem \ref{Pro1}. 
We aim to determine the ST for a given LQR problem \eqref{1},\eqref{lambda distribution},\eqref{2}, i.e., given $A,B,Q,R,q$. Because ST may depend on all these parameters, it is appropriate to denote it as $\bar{\delta}(A,B,Q,R,q)$. And we abbreviate it as $\bar{\delta}(q)$, since $q$ is the essential cause of the estimation error.
In fact, it is challenging to analytically express the ST $\bar{\delta}(q)$ in terms of the given parameters $A,B,Q,R,q$, even in the scalar case. Therefore, we will provide a lower bound for the ST. This lower bound can be expressed explicitly based on the given parameters and ensures that $\hat{K}$ is stabilizing when $q-\hat{q}$ is smaller than this lower bound. We first show a useful lemma.
\begin{lemma}\label{lemma 1}\rm
    With Assumption \ref{assump1} and $\hat{q}<q_c$ holding, given $\hat{q}<q$, then the corresponding matrices $P$ and $\hat{P}$ in \eqref{RE1} and \eqref{RE2} satisfy $\hat{P}\le P$.
\end{lemma}
\begin{proof}
    Let's define $g(X)\triangleq Q+A^{\top}XA-(1-q)A^{\top}XB(R+B^{\top}XB)^{-1}B^{\top}XA$, $\hat{g}(X)\triangleq Q+A^{\top}XA-(1-\hat{q})A^{\top}XB(R+B^{\top}XB)^{-1}B^{\top}XA$, where $X\in\mathbb{R}^{n\times n}$. 
    Given any matrix $\mathcal{P}_0>0$ as the initial matrix and denoting  $P_{t+1}=g(P_t),~\hat{P}_{t+1}=g(\hat{P}_t)$, we know from \cite{kalman2005} that $P=\lim_{t\to\infty}P_t=\lim_{t\to\infty}g^t(\mathcal{P}_0)$, $\hat{P}=\lim_{t\to\infty}\hat{P}_t=\lim_{t\to\infty}\hat{g}^t(\mathcal{P}_0)$.

    When $\hat{q}<q$, i.e. $1-\hat{q}>1-q$, we know from c) of Lemma 1 in \cite{kalman2005} that $\hat{P}_1=\hat{g}(\mathcal{P}_0)\le g(\mathcal{P}_0)=P_1$.
    And from d) of Lemma 1 in \cite{kalman2005}, we have $       \hat{g}(\hat{P}_1)\le g(\hat{P}_1)\le g(P_1)$, which shows $\hat{P}_2\le P_2$.
By induction, there is $\hat{P}_t\le P_t$ for $\forall t$. Let $t\to\infty$, we have $\hat{P}\le P$.
\end{proof}

Then, the lower bound of ST is given as follows.

\begin{theorem}\rm\label{Th2}
    With Assumption \ref{assump1} and $\hat{q}<q_c$ holding, for the given system \eqref{1}, the ST satisfies 
    \begin{equation}\label{th2_35}
        \bar{\delta}(q)\ge \frac{\lambda_{\min}\{Q^{\frac{1}{2}}(A^{\top}P^2A)^{-1}Q^{\frac{1}{2}}\}}{c_1},
    \end{equation}
    where $c_1=\frac{\lambda_{\max}(BB^{\top})}{\lambda_{\min}(R+B^{\top}P_0B)}$, $P$ is the positive definite solution of \eqref{RE1}, $P_0$ is the positive definite solution of the following Riccati equation
\begin{equation}\label{th2_35(2)}
  P_0=Q+A^{\top}P_0A-A^{\top}P_0B(R+B^{\top}P_0B)^{-1}B^{\top}P_0A.
\end{equation}
\end{theorem}
\begin{proof}    
For any $q$ and $\hat{q}$ satisfying $0<q-\hat{q}<\lambda_{\min}\{Q^{\frac{1}{2}}(A^{\top}P^2A)^{-1}Q^{\frac{1}{2}}\}/{c_1}$,
we know $P_0\le\hat{P}\le P$ from Lemma \ref{lemma 1} and $Q-(q-\hat{q})c_1A^{\top}P^2A> 0$, where $\hat{P}$ satisfies \eqref{RE2}. Then, we have
    \begin{equation}
        \begin{aligned}
            &Q\!+\!(1\!-\!q)\hat{K}^{\top}\!R\hat{K}\!-\!(q\!-\!\hat{q})A^{\top}\!\hat{P}B(\!R\!+\!B^{\top}\!\hat{P}\!B)^{-1}\!B^{\top}\!\hat{P}\!A\\
            \ge& Q-(q-\hat{q})\frac{\lambda_{\max}(BB^{\top})}{\lambda_{\min}(R+B^{\top}P_0B)}A^{\top}\hat{P}^2A\\
            \ge& Q-(q-\hat{q})c_1A^{\top}P^2A> 0.
        \end{aligned}
    \end{equation}
 According to Theorem \ref{Pro1}, \eqref{hat K} mean-square stabilizes \eqref{1}. 
\end{proof}
\begin{remark}\label{the impact of parameter}
    If $q$ increases, the upper bound on $q-\hat{q}$ becomes more restrictive due to the increase of $P$ with respect to $q$. In particular, if $q$ approaches $q_c$, the upper bound will approach $0$. Therefore, in poorer communication channels, the estimate should be more accurate to ensure $\hat{K}$ is stabilizing. As shown in Fig. \ref{th3 stabilizing region}, the closer $q$ is to $q_c$, the smaller the range of stabilizable $\hat{q}$, where the blue area is the stabilizing region of $\hat{q}$. 
\end{remark}
\begin{figure}[h]
	\centering
	\subfigure[Stabilizing region for \Cref{ex 1}.]{\includegraphics[width=0.23\textwidth]{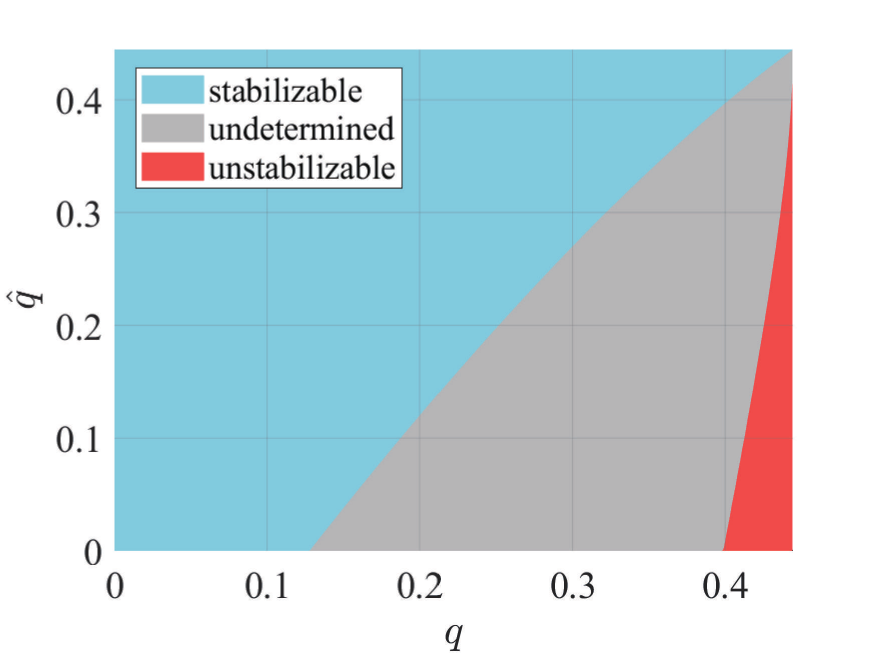}
		\label{th3 scalar hat_q--q}}
	\subfigure[Stabilizing region for \Cref{exa 2}.]{
		\label{th3 vector hat_q--q}
		\includegraphics[width=0.23\textwidth]{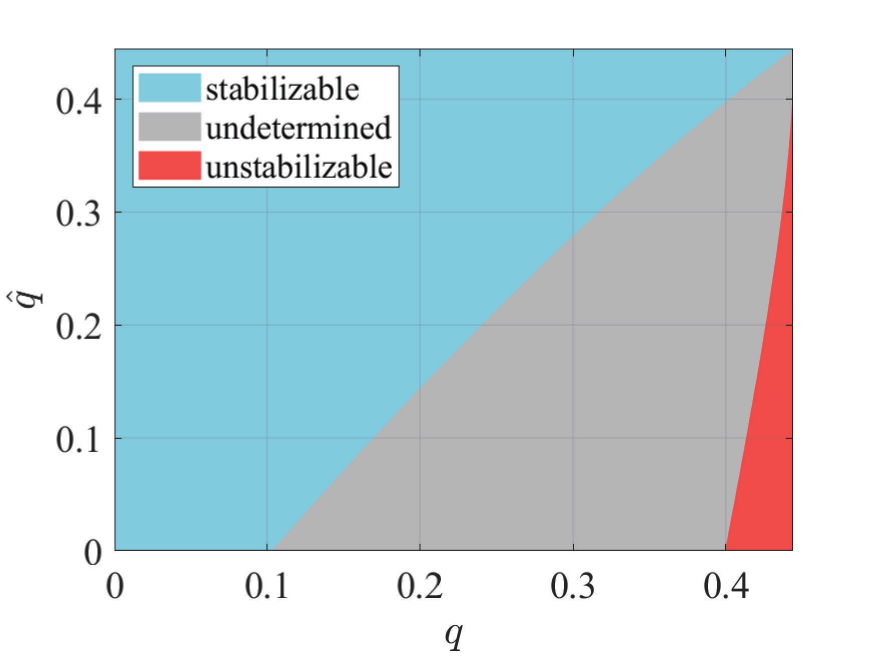}}
	\caption{The stabilizing regions of $\hat{q}$ distinguished by \Cref{Th2} are represented in blue, where \Cref{ex 1}, \Cref{exa 2} are given in \Cref{section6}.}
	\label{th3 stabilizing region}
\end{figure}

From Theorem \ref{Th2}, we also know that when $q$ is smaller than the lower bound of the ST, \eqref{hat K} designed with any $\hat{q}\in[0,q_c)$ mean-square stabilizes the system \eqref{1}, as shown below.
\begin{corollary}\label{cor2}\rm
    With Assumption \ref{assump1} and $\hat{q}<q_c$ holding, for any $\hat{q}\in[0,q_c)$, \eqref{hat K} mean-square stabilizes the system \eqref{1} if $q$ satisfies $q< \lambda_{\min}\{Q^{\frac{1}{2}}(A^{\top}P^2A)^{-1}Q^{\frac{1}{2}}\}/{c_1}$, where the definitions of $c_1,P,P_0$ are the same as those in Theorem \ref{Th2}.
\end{corollary}



Theorem \ref{Th2} provides a lower bound of the ST for $n$-dimensional linear systems. Moreover, tighter lower bounds can be obtained for scalar systems and $n$-dimensional systems with invertible input matrix. 

\subsection{Tailored Results for Specific Classes of Systems}

First, the lower bound of ST is given for scalar systems.

\begin{theorem}\rm\label{scalar th3}
    With Assumption \ref{assump1} and $\hat{q}<q_c$ holding, for the given system \eqref{1}, supposing $n=m=1$, the ST satisfies 
    \begin{equation}\label{scalar low bound}
        \bar{\delta}(q)\ge \frac{Q(R+B^2P)}{A^2B^2P^2}+\frac{(1-q)R}{R+B^2P},
    \end{equation}
    where $P$ is the positive definite solution of \eqref{RE1}.
\end{theorem}

\begin{proof}
Through algebraic manipulation, we derive that \eqref{scalar condition} in Theorem \ref{thm n and s scalar} is equivalent to $q-\hat{q}<\frac{Q(R+B^2\hat{P})}{A^2B^2\hat{P}^2}+\frac{(1-q)R}{R+B^2\hat{P}}$.
Define $f(X)\triangleq\frac{Q(R+B^2X)}{A^2B^2X^2}+\frac{(1-q)R}{R+B^2X}$
where $X\in\mathbb{R}_+$. Since $Q,R>0$, $f(X)$ is monotonically decreasing with respect to $X\in\mathbb{R}_+$. From Lemma \ref{lemma 1}, we know $f(\hat{P})\ge f(P)$ if $\hat{q}<q$. Hence, when $0\le q-\hat{q}\le \frac{Q(R+B^2P)}{A^2B^2P^2}+\frac{(1-q)R}{R+B^2P}$, \eqref{scalar condition} holds.
\end{proof}
\begin{remark}\rm
    In the scalar case, Theorem \ref{scalar th3} is less conservative than Theorem \ref{Th2}. In fact, for the scalar case, the right-hand side (RHS) of \eqref{th2_35} can be rewritten as $\frac{Q(R+B^2P_0)}{A^2B^2P^2}$ which is clearly smaller than the RHS of \eqref{scalar low bound}.
Moreover, from Theorem \ref{scalar th3}, when $q<\frac{Q(R+B^2P)}{A^2B^2P^2}+\frac{(1-q)R}{R+B^2P}$, \eqref{hat K} designed with any $\hat{q}\in[0,1/\rho(A)^2)$ mean-square stabilizes the system \eqref{1}.
\end{remark}


Next, the lower bound of ST is provided for $n$-dimensional systems with invertible $B$.
\begin{theorem}\rm\label{Th 3}
    With Assumption \ref{1} and $\hat{q}<q_c$ holding, for the given system \eqref{1}, supposing $B$ is invertible, the ST satisfies
\begin{equation}\label{32}
    \bar{\delta}(q)\ge \lambda_{\min}\{\Xi(A^{\top}PA)^{-1}\Xi\},
\end{equation}
where $\Xi=(Q+(1-q)c_2A^{\top}P_0^2A)^{\frac{1}{2}}$,
$c_2=\frac{\lambda_{\min}(R)\lambda_{\min}(BB^{\top})}{\lambda_{\max}(R+B^{\top}PB)^2}$, $P$, $P_0$ are defined as in Theorem \ref{Th2}. 

\end{theorem}
\begin{proof}
    When $0<q-\hat{q}<\lambda_{\min}\{\Xi(A^{\top}PA)^{-1}\Xi\}$, we have
\begin{align}
    &Q\!+\!(1\!-\!q)\hat{K}^{\top}\!R\hat{K}\!-\!(q\!-\!\hat{q})A^{\top}\hat{P}B(R\!+\!B^{\top}\hat{P}B)^{-1}B^{\top}\hat{P}A\notag\\
    \ge&Q+(1-q)\frac{\lambda_{\min}(R)\lambda_{\min}(BB^{\top})}{\lambda_{\max}(R+BPB^{\top})^2}A^{\top}\hat{P}^2A-(q-\hat{q})\times\notag\\
    &A^{\top}\hat{P}[(BR^{-1}B^{\top})^{-1}+\hat{P}]^{-1}\hat{P}A\notag\\
    \overset{(a)}{\ge}&Q+(1-q)c_2A^{\top}P_0^2A-(q-\hat{q})A^{\top}PA>0\label{B proof result},
\end{align}    
where $(a)$ is derived from $\hat{P}^{\frac{1}{2}}[(BR^{-1}B^{\top}+\hat{P})]^{-1}\hat{P}^{\frac{1}{2}}\le\mathcal{I}$ and Lemma \ref{lemma 1}. From Theorem \ref{Pro1}, \eqref{hat K} is stabilizing.
\end{proof}

    \begin{remark}
     The characteristics in Remark \ref{the impact of parameter} still hold for the tailored results, for similar reasons. In a poorer communication channel, as $q$ increases, the stabilizing ranges of $\hat{q}$ in Theorem \ref{scalar th3}, \ref{Th 3} become more restrictive, as shown in Fig. \ref{th4 and th5 stabilizing region}. 
    \end{remark}
    \begin{figure}[h]
    	\centering
    	\subfigure[Stabilizing region for \Cref{ex 1} according to \Cref{scalar th3}.]{\includegraphics[width=0.23\textwidth]{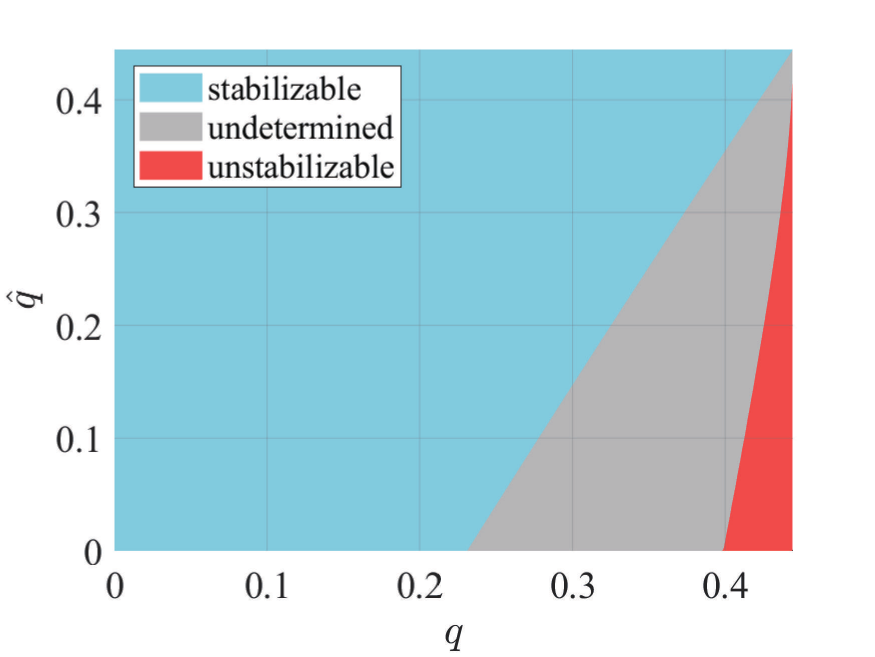}
    		\label{th4 hat_q--q}}
    	\subfigure[Stabilizing region for \Cref{exa 2} according to \Cref{Th 3}.]{
    		\label{th5 hat_q--q}
    		\includegraphics[width=0.23\textwidth]{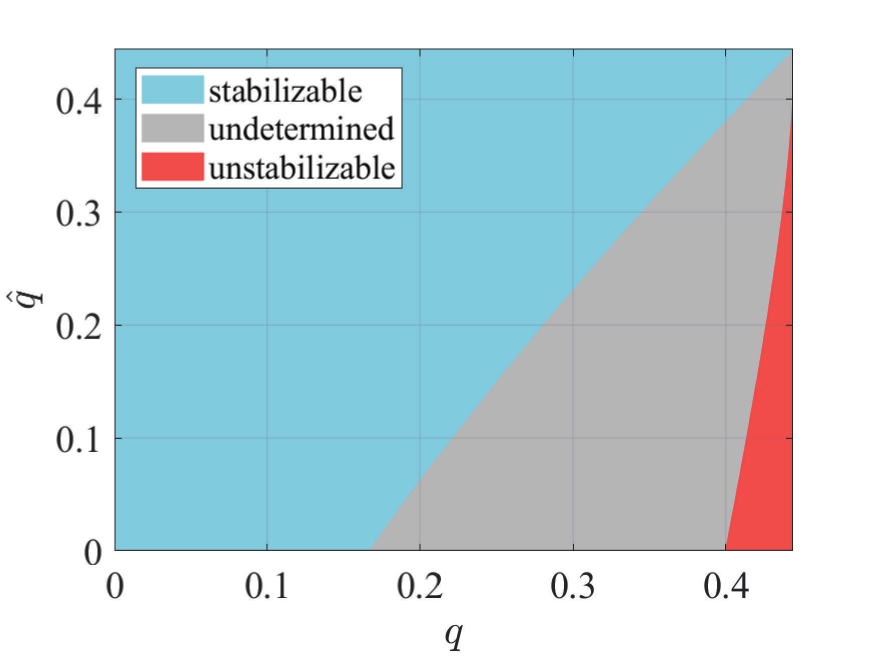}}
    	\caption{The stabilizing regions of $\hat{q}$ are represented in blue.}
    	\label{th4 and th5 stabilizing region}
    \end{figure}
    
    \begin{remark}\rm
        For systems with invertible $B$, Theorem \ref{Th 3} is less conservative than Theorem \ref{Th2}. If $q\!-\!\hat{q}\!<\!\lambda_{\min}\!\{\!Q^{\frac{1}{2}}\!(\!A^{\top}\!P^2\!A)^{-1}Q^{\frac{1}{2}}\!\}/{c_1}$, then $c_1(q-\hat{q})A^{\top}P^2A< Q$. Since $A^{\top}\!P\!A\!<\! A^{\top}\!P\!B(R\!+\!B^{\top}\!P\!B)^{-1}\!B^{\top}\!P\!A\!\le\! c_1A^{\top}\!P^2\!A$ when $B$ is invertible, we know $(q-\hat{q})A^{\top}PA< Q+(1-q)c_2A^{\top}P_0^2A$. So $q-\hat{q}<\lambda_{\min}\{\Xi(A^{\top}PA)^{-1}\Xi\}$.
Additionally, Theorem \ref{Th 3} also applies to scalar systems, but it is more conservative than Theorem \ref{scalar th3} in the scalar case. It can be verified that, when $n=m=1$, $
               \lambda_{\min}\{\Xi(A^{\top}PA)^{-1}\Xi\}
               =\frac{Q+(1-q)RA^2B^2P_0^2(R+B^{2}P)^{-2}}{A^2P}
               \le\frac{Q(B^{-2}R+P)+(1-q)RA^2P^2(R+B^2P)^{-1}}{A^2P^2}
               =\frac{Q(R+B^2P)}{A^2B^2P^2}+\frac{(1-q)R}{R+B^2P}$,
        where $\Xi$, $P_0$ are defined as in Theorem \ref{Th 3}. Moreover, from Theorem \ref{Th 3}, when $q<\lambda_{\min}\{\Xi(A^{\top}PA)^{-1}\Xi\}$, \eqref{hat K} designed with any $\hat{q}\in[0,\frac{1}{\rho(A)^2})$ mean-square stabilizes the system \eqref{1}.
    \end{remark}

In this section, we present upper bounds of $q-\hat{q}$ for $\hat{K}$ to be stabilizing. Essentially, the presence of non-zero estimation error stems from the fact that only finite samples are utilized to estimate the loss probability. Next, we will investigate the sample complexity for $\hat{K}$ to be stabilizing.

\section{Samples complexity analysis}\label{section4}
Based on the strong law of large numbers (Theorem 5.18 of \cite{statistic book}), as $N_q$ tends to $+\infty$, $\hat{q}$ converges to $q$ almost surely. Yet in practice, only finite samples are available. Hence, we are concerned with sample complexity.
The below lemma is a variation of Hoeffding's inequality, which establishes a bridge between sample size and estimation error.
\begin{lemma}\rm(Theorem 4.5 of \cite{statistic book})\label{lemma0}
    Given a sample sequence $\{\lambda_i,i=1,\dots,N_q\}$ and letting $\hat{q}=\frac{1}{N_q}\sum_{i=1}^{N_q}(1-\lambda_i)$. Then for any $\beta\in(0,1)$, it holds that 
$\mathcal{P}\{|q-\hat{q}|\le \Delta(N_q,\beta)\}\ge1-\beta$, where $\Delta(N_q,\beta):=\sqrt{\frac{\log(2/\beta)}{2N_q}}$.
\end{lemma}




In fact, for a certain $N_q$, if the upper bound on the estimation error is smaller than the ST, $\hat{K}$ is stabilizing. Based on this idea, we give the next lemma regarding the sample complexity.

\begin{lemma}\rm\label{pro2}
    Under Assumption \ref{assump1}, given a sequence of samples $\{\lambda_i,i=1,\dots,N_q\}$ such that $\hat{q}<q_c$, \eqref{hat K} mean-square stabilizes the system \eqref{1} with probability at least $1-\beta,~\forall \beta\in(0,1)$, if $N_q$ satisfies
    \begin{equation}\label{pro2_28}
        N_q>\frac{\log(2/\beta)}{2\bar{\delta}(q)^2},
    \end{equation}
    where $\bar{\delta}(q)$ is the ST for the given system \eqref{1}.
\end{lemma}
\begin{proof}
According to Lemma \ref{lemma0}, we know that $\mathcal{P}\{-\Delta(N_q,\beta){\le}q-\hat{q}\le \Delta(N_q,\beta)\} \ge1-\beta$. Then, based on \eqref{pro2_28}, we have $\Delta(N_q,\beta)<\bar{\delta}(q)$. Therefore, $q-\hat{q}<\bar{\delta}(q)$ holds with probability at least $1-\beta$.
\end{proof}
\begin{remark}
	\eqref{pro2_28} in this paper appears similar to (21) in \cite{unknown channel 7}, but their purposes are quite different. (21) in \cite{unknown channel 7} ensures that the worst-case loss probability is less than $q_c$ such that the ``worst-case" CE controller exists. 
	In contrast, under the premise that $\hat{q}<q_c$, this paper provides \eqref{pro2_28} as a sufficient condition for the CE controller synthesized from $\hat{q}$ to be stabilizing in a probabilistic sense.
\end{remark}


The RHS of \eqref{pro2_28} represents the sample complexity for $\hat{K}$ to be stabilizing. Evidently, as $\bar{\delta}(q)$ increases, the sample complexity decreases. By integrating the analysis of $\bar{\delta}(q)$ from the previous section, we derive an upper bound on the sample complexity for $n$-dimensional systems in terms of the system parameters $A,B,Q,R,q$, as formalized in the next theorem. 
\begin{theorem}\rm\label{Th5}
Under Assumption \ref{1}, given a sequence of samples $\{\lambda_i,i=1,\dots,N_q\}$ such that $\hat{q}<q_c$, \eqref{hat K} mean-square stabilizes the system \eqref{1} with probability at least $1-\beta$, if $N_q$ satisfies
\begin{equation}\label{th5_36}
        N_q> \frac{c_1^2\log(2/\beta)}{2\lambda_{\min}\{Q^{\frac{1}{2}}(A^{\top}P^2A)^{-1}Q^{\frac{1}{2}}\}^2},
\end{equation}
where $\beta\in(0,1)$, and $c_1,P$ are defined as in Theorem \ref{Th2}.
\end{theorem}

Substituting the lower bound of $\bar{\delta}(q)$ from Theorem \ref{Th2} into \eqref{pro2_28}, Theorem \ref{Th5} can be directly obtained.
The RHS of \eqref{th5_36} serves as an upper bound of the sample complexity. Furthermore, when $q$ is sufficiently small, the sample size required for stability may be less than \eqref{th5_36}. Based on Corollary \ref{cor2}, when $q\!<\!{\lambda_{\min}\{Q^{\frac{1}{2}}(A^{\top}\!P^2\!A)^{-1}Q^{\frac{1}{2}}\}}/{c_1}$, $\hat{K}$ stabilizes the system even if $N_q=0$. However, when $q$ is large and approaches $q_c$, the sample complexity will increase rapidly and tend toward $+\infty$, as $P$ increases unboundedly \cite{kalman2005}, as shown in Fig. \ref{th6 sample complexity}.
\begin{figure}[h]
	\centering
	\subfigure[Sample complexity of \Cref{ex 1} based on \eqref{th5_36}.]
	{\includegraphics[width=0.23\textwidth]{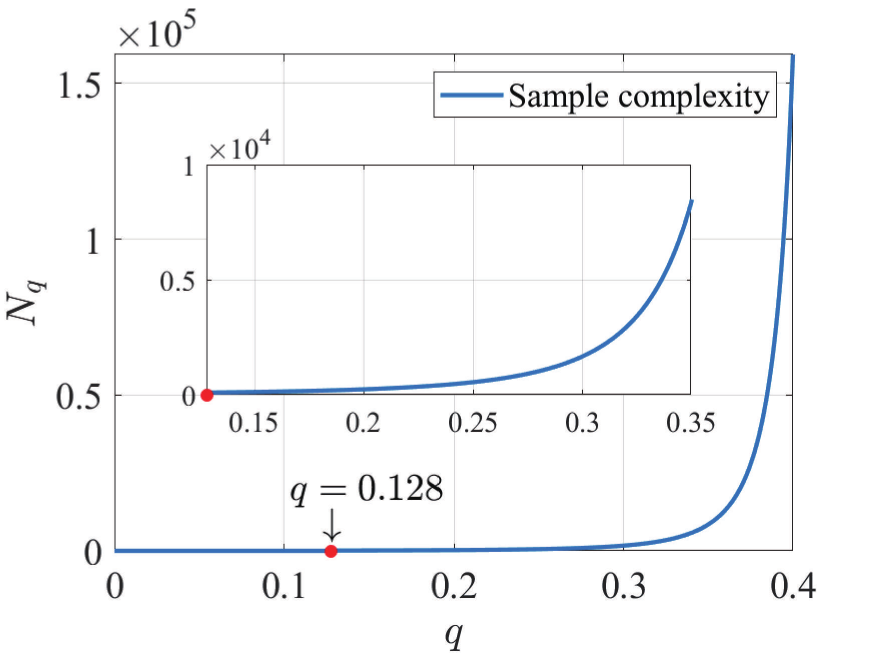}
		\label{th6 scalar sample complexity}}
	\subfigure[Sample complexity of \Cref{exa 2} based on \eqref{th5_36}.]{
		\label{th6 vector sample complexity}
		\includegraphics[width=0.23\textwidth]{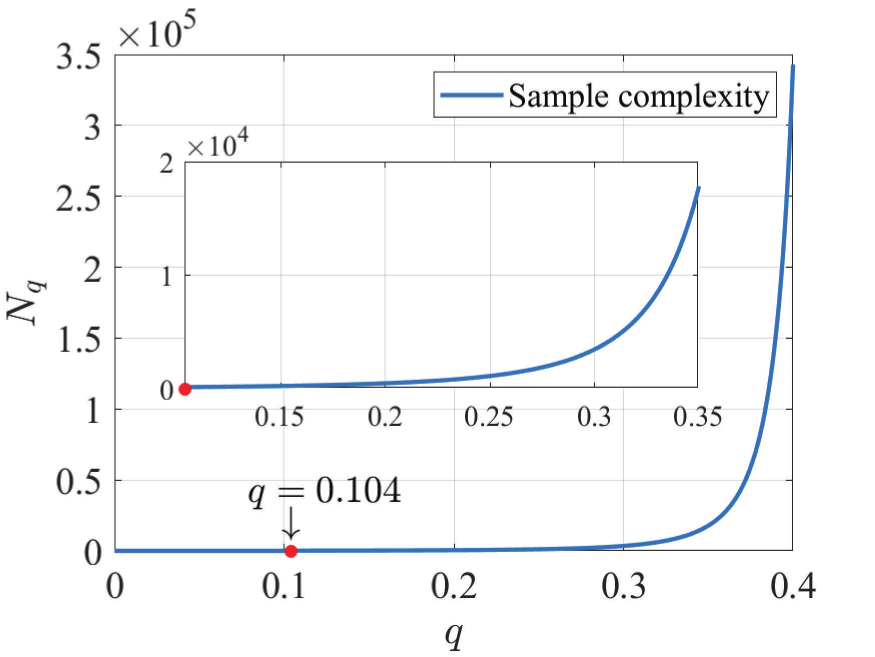}}
	\subfigure[Sample complexity of \Cref{ex 1} based on \eqref{scalar sample complexity}.]
		{\includegraphics[width=0.23\textwidth]{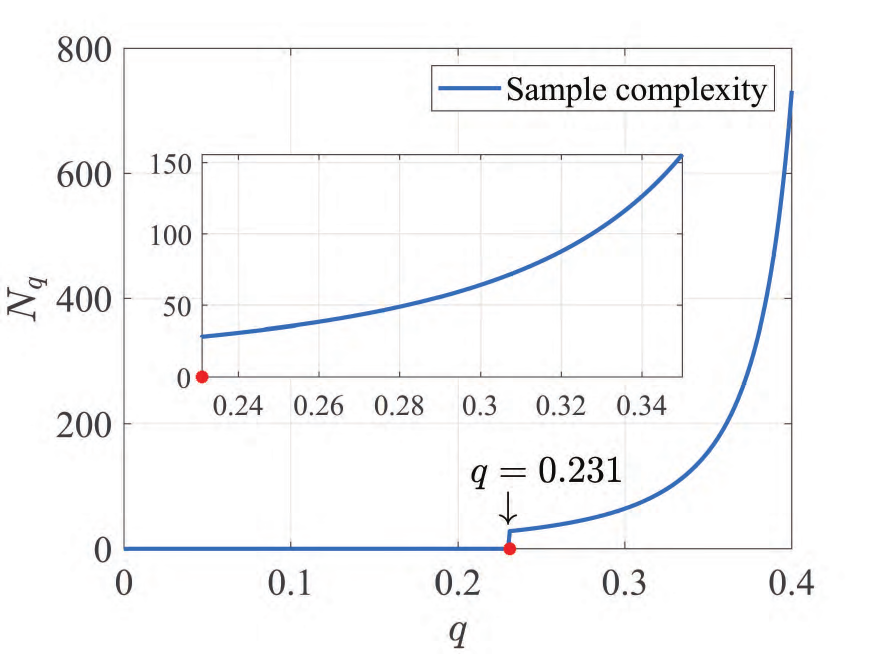}
			\label{sample complexity tailored for scalar}}
	\subfigure[Sample complexity of \Cref{exa 2} based on \eqref{invertible B sample complexity}.]{
			\label{sample complexity tailored for vector}
			\includegraphics[width=0.23\textwidth]{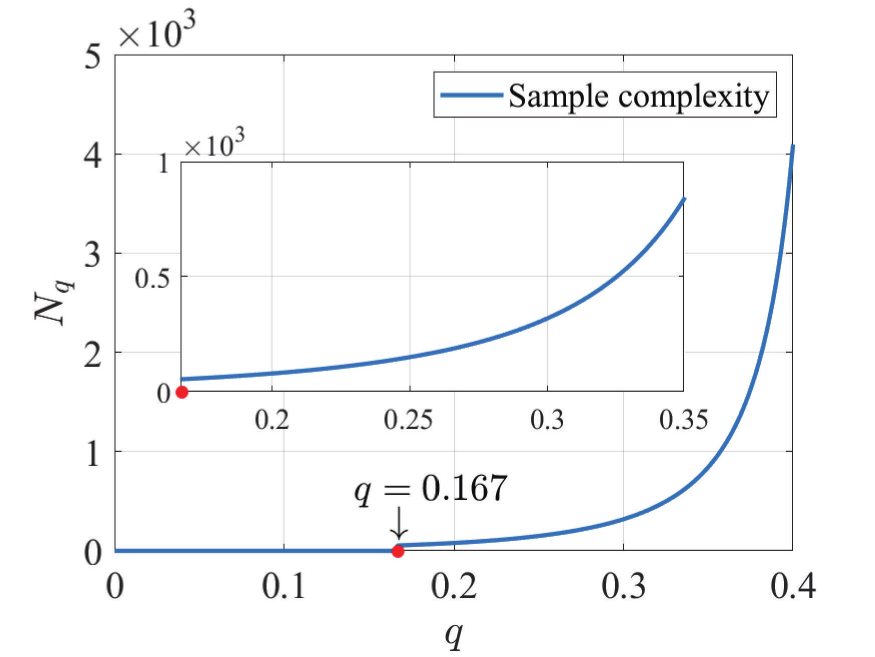}}
	\caption{Sample complexity for \eqref{hat K} to be stabilizing with $\beta\!\!=\!\!0.1$.}
	\label{th6 sample complexity}
\end{figure}

Additionally, under the same assumptions as in Theorem \ref{Th5}, for scalar systems, $\hat{K}$ mean-square stabilizes the system with probability at least $1-\beta$ if $N_q$ satisfies
\begin{equation}\label{scalar sample complexity}
        N_q>\frac{\log(2/\beta)A^4B^4P^4(R+B^2P)^2}{[Q(R+B^2P)^2+(1-q)RA^2B^2P^2]^2};
    \end{equation}
for systems with invertible $B$, $\hat{K}$ stabilizes the system with probability at least $1-\beta$ if $N_q$ satisfies 
\begin{equation}\label{invertible B sample complexity}
            N_q> \frac{\log(2/\beta)}{\lambda_{\min}\{\Xi(A^{\top}PA)^{-1}\Xi\}^2}.
        \end{equation}
The results \eqref{scalar sample complexity} and \eqref{invertible B sample complexity} can be proven by combining Lemma \ref{pro2} with Theorem \ref{scalar th3} and \ref{Th 3}, respectively. The sample complexity curves based on \eqref{scalar sample complexity} and \eqref{invertible B sample complexity} are shown in Fig. \ref{sample complexity tailored for scalar}, \ref{sample complexity tailored for vector}.

Since ST and sample complexity depend on the unknown $q$, they cannot be directly used to determine whether $\hat{K}$ is stabilizing in practice. Therefore, from a practical perspective, a sufficient condition, independent of $q$, is provided to verify whether $\hat{K}$ is stabilizing with a certain probability.

\begin{theorem}\label{Th1}\rm
    Under Assumption \ref{assump1}, given a sequence of samples $\{\lambda_i,i\!=\!1,\dots,N_q\}$ such that $\hat{q}<q_c$, \eqref{hat K} mean-square stabilizes system \eqref{1} with probability at least $1-\beta$, if 
    \begin{equation}\label{Th1_23}
        \bar{q}\ge \hat{q}+\Delta(N_q,\beta),
    \end{equation}
    where $\beta\in(0,1)$, $\Delta(N_q,\beta)$ is defined in Lemma \ref{lemma0}, $\bar{q}$ is the optimal value of the semi-definite programming (SDP)
    \begin{subequations}\label{SDP}
        \begin{alignat}{2}\label{SDP1}
            \max_{x\in[0,1]} \quad & x\\
            {\rm s.t.} \quad & x\ge \hat{q}\label{SDP2}\\
            &Q+(1-x)\hat{K}^{\top}R\hat{K}-(x-\hat{q})\times\notag\notag\\
            &A^{\top}\hat{P}B(R+B^{\top}\hat{P}B)^{-1}B^{\top}\hat{P}A>0,\label{SDP3}
        \end{alignat}
    \end{subequations}
    where $\hat{K}$, $\hat{P}$ are defined in \eqref{hat K}, \eqref{RE2}, respectively.
\end{theorem}
\begin{proof}
    According to Lemma \ref{lemma0}, we have $\mathcal{P}\{-\Delta(N_q,\beta){\le}q-\hat{q}\le \Delta(N_q,\beta)\}\ge1-\beta.$
Combining with \eqref{Th1_23}, we know $q\le \bar{q}$ with probability at least $1-\beta$.
Hence, there is $Q\!+\!(1\!-\!q)\hat{K}^{\top}\!R\hat{K}\!-\!(q\!-\!\hat{q})A^{\top}\!\hat{P}B(R\!+\!B^{\top}\!\hat{P}\!B)^{-1}\!B^{\top}\!\hat{P}\!A
>\!Q\!+\!(1\!-\!\bar{q})\hat{K}^{\top}\!R\hat{K}\!-\!(\bar{q}-\hat{q})A^{\top}\!\hat{P}B(R\!+\!B^{\top}\!\hat{P}\!B)^{-1}B^{\top}\!\hat{P}\!A$
    with probability at least $1-\beta$. Then, from the constraint \eqref{SDP3}, we know that 
$Q+(1-q)\hat{K}^{\top}R\hat{K}-(q-\hat{q})A^{\top}\hat{P}B(R+B^{\top}\hat{P}B)^{-1}B^{\top}\hat{P}A>0$
    with probability at least $1-\beta$. According to Theorem \ref{Pro1}, \eqref{hat K} mean-square stabilizes the system \eqref{1} with probability at least $1-\beta$.
\end{proof}

\begin{remark}
Theorem \ref{Th1} provides a sufficient condition for $\hat{K}$ to mean-square stabilize the system \eqref{1}. Given a $\hat{q}$ estimated from $N_q$ samples, if \eqref{Th1_23} holds, then $\hat{K}$ is stabilizing with probability at least $1-\beta$. Since \eqref{Th1_23} is independent of $q$, it can be directly applied in practice. It is important to emphasize that $\bar{q}$ is only used to verify the validity of \eqref{Th1_23} and is not involved in constructing the controller. Therefore, even a large $\bar{q}$ does not introduce additional conservatism into this theorem. Moreover, the SDP \eqref{SDP} is always feasible. For instance, $x=\hat{q}$ is a feasible solution as it clearly satisfies \eqref{SDP2} and \eqref{SDP3}. 
\end{remark}

A numerical example is provided to illustrate the applicability of Theorem \ref{Th1}. Consider the two-dimensional system $(A,B)$ given in Example \ref{exa 2} with initial value $x_0=[0.9325;1.1616]$. Assume $q=0.2$ which satisfies Assumption \ref{assump1} and is unknown to the controller. Using $300$ samples, we obtain an estimated probability $\hat{q}=0.1633$. Substituting $\hat{q}=0.1633$ into \eqref{SDP2} and \eqref{SDP3}, the solution to \eqref{SDP} is $\bar{q}=0.4181$. It can be verified that \eqref{Th1_23} holds when $\beta=0.01$, and then we know $\hat{K}$ designed with $\hat{q}=0.1633$ is stabilizing with probability at least $0.99$. The system trajectory of Example \ref{exa 2} under $\hat{K}$ is displayed in Fig. \ref{stabilize system trajectory}.

\begin{figure}
    \centering
    \subfigure[$\hat{K}$ designed with $\hat{q}=0.1633$ is stabilizing, which is estimated by $300$ samples.]{
    \includegraphics[width=0.45\linewidth]{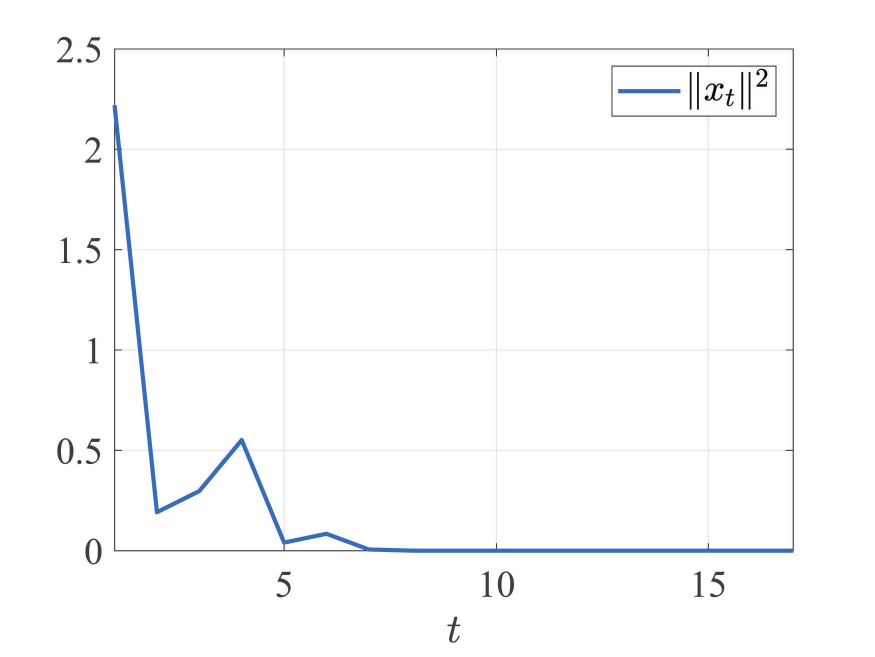}
    \label{stabilize system trajectory}}
    \subfigure[The optimality gap between $\hat{K}$ and $K$, i.e., $J(x_0,\hat{\mathbf{u}})-J^*(x_0)$.]{
    \includegraphics[width=0.45\linewidth]{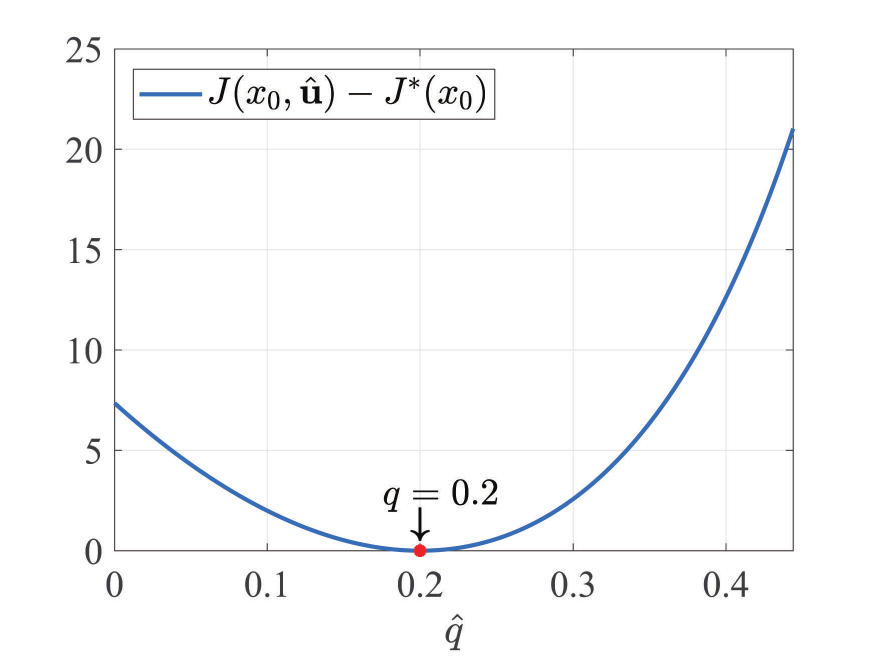}
    \label{gap cost}
    }
    \caption{The trajectory and gap between $J(x_0,\hat{\mathbf{u}})$ and $J^*(x_0)$ for the system given in Example \ref{exa 2} with $q=0.2$.}
\end{figure}

\section{Performance analysis}\label{section5}
The stability issue under unknown packet drop probability has been discussed earlier. Next, the performance of $\hat{K}$ is also an important concern. In this section, we analytically quantify the sub-optimality of $\hat{K}$ in terms of the estimation error and prove that the optimality gap converges to $0$ as $\hat{q}$ approaches to $q$. First, the continuity of $\hat{P}$ with respect to $\hat{q}$ is given below.
\begin{lemma}\label{hat_P limits}
    With Assumption \ref{assump1} and $\hat{q}<q_c$ holding, if $\hat{q}$ converges to $q$, then the corresponding matrix $\hat{P}$ converges to $P$, i.e., $\lim_{\hat{q}\to q}\hat{P}=P$, where $\hat{P}$, $P$ are the positive definite solutions of \eqref{RE1} and \eqref{RE2} respectively.
\end{lemma}
\begin{proof}
    When $\hat{q}<q$ and $\hat{q}$ increases towards $q$, based on Lemma \ref{lemma 1}, $\hat{P}$ is monotonically increasing and bounded. Therefore, there exists a positive definite matrix $\hat{P}^-$ such that $\lim_{\hat{q}{\to}q^-}\hat{P}=\hat{P}^-$.
Then, taking the left limit with respect to $\hat{q}$ on both side of \eqref{RE2}, we have 
        $\hat{P}^-\!=\!Q\!+\!A^{\top}\!\hat{P}^-\!\!A\!-\!(1\!-q)A^{\top}\!\hat{P}^-\!B(\!R\!+\!B^{\top}\!\hat{P}^-\!B)^{-1}\!B^{\top}\!\hat{P}^-\!\!A$
    which shows $\hat{P}^-$ is the positive definite solution of \eqref{RE1}. We have $\hat{P}^-=P$ due to the uniqueness of the solution of \eqref{RE1}. Similarly, when $\hat{q}>q$ and $\hat{q}$ decreases towards $q$, there also exists a positive definite matrix $\hat{P}^+$ satisfying $\lim_{\hat{q}{\to}q^+}\hat{P}=\hat{P}^+$ and it can be demonstrated that $\hat{P}^+=P$. In summary, $\lim_{\hat{q}{\to}q}\hat{P}=P$.
\end{proof}

Then, the next theorem analytically expresses the optimality gap as a linear combination of $q-\hat{q}$ and $\hat{P}-P$.

\begin{theorem}\label{performance analysis}
    With Assumption \ref{assump1} and $\hat{q}<q_c$ holding, when \eqref{hat K} mean-square stabilizes the system \eqref{1}, the gap between $J(x_0,\mathbf{\hat{u}})$ and the optimal cost $J^*(x_0)$ satisfies
    \begin{equation}\label{performance gap}
        \begin{aligned}
            J(x_0,\hat{\mathbf{u}})-J^*(x_0)=&{\rm tr}\big\{(q-\hat{q})X_{\hat{K}}+(\hat{P}-P)X_0\big\}
        \end{aligned}
    \end{equation}
    where $X_{\hat{K}}=A^{\top}\hat{P}B(R+B^{\top}\hat{P}B)^{-1}B^{\top}\hat{P}A\mathbb{E}\{\sum_{t=0}^{\infty}x_tx_t^{\top}\}$, $X_0=\mathbb{E}(x_0x_0^{\top})$ are bounded positive semi-definite matrices.
Moreover, as $\hat{q}$ converges to $q$, the gap $J(x_0,\mathbf{\hat{u}})-J^*(x_0)$ converges to $0$.
\end{theorem}
\begin{proof}
    When \eqref{hat K} stabilizes the system \eqref{1}, the cost $J(x_0,\mathbf{\hat{u}})=\sum_{t=0}^{\infty}\mathbb{E}\Big\{x_t^{\top}Qx_t+\lambda_t\hat{u}_tR\hat{u}_t\Big\}$ is finite, and $\lim_{t\to\infty}\mathbb{E}\{x_{t+1}^{\top}\hat{P}x_{t+1}\}=0$. Therefore, 
    \begin{align}\label{J(hat u)}
        &J(x_0,\mathbf{\hat{u}})
        =\lim_{t\to\infty}\sum_{n=0}^t\mathbb{E}\{x_n^{\top}Qx_n+\lambda_n\hat{u}_n^{\top}R\hat{u}_n+x_{t+1}^{\top}\hat{P}x_{t+1}\}\notag\\
        \overset{(a)}{=}&\lim_{t\to\infty}\Bigg[\sum_{n=0}^{t-1}\mathbb{E}\{x_n^{\top}Qx_n+\lambda_n\hat{u}_n^{\top}R\hat{u}_n+x_t^{\top}\hat{P}x_t\}\notag\\
        &+(q-\hat{q})\mathbb{E}\{x_t^{\top}A^{\top}\hat{P}B(R+B^{\top}\hat{P}B)^{-1}B^{\top}\hat{P}Ax_t\}\Bigg]\notag\\
        \overset{(b)}{=}&\lim_{t\to\infty}\Bigg[(q-\hat{q})\sum_{n=0}^{t}\mathbb{E}\{x_n^{\top}A^{\top}\hat{P}B(R+B^{\top}\hat{P}B)^{-1}B^{\top}\hat{P}Ax_n\}\notag\\
        +&\mathbb{E}\{x_0^{\top}\hat{P}x_0\}\Bigg]={\rm tr}\big\{(q-\hat{q})X_{\hat{K}}+\hat{P}X_0\big\},
    \end{align}
    where $(a)$ follows from $\hat{u}_t=-(R+B^{\top}\hat{P}B)^{-1}B^{\top}\hat{P}Ax_t$; $(b)$ follows from substituting $\hat{u}_{t-1},\hat{u}_{t-2},\dots,\hat{u}_0$ in sequence. Because $\hat{K}$ mean-square stabilizes the system \eqref{1}, $\mathbb{E}\{\sum_{t=0}^{\infty}x_tx_t^{\top}\}$ is bounded. Based on \eqref{J(hat u)} and $J^*(x_0)={\rm tr}\{PX_0\}$, we obtain \eqref{performance gap}. Moreover, since $X_{\hat{K}},X_{0}$ are bounded and $\lim_{\hat{q}{\to}q}\hat{P}=P$ (given by Lemma \ref{hat_P limits}), there is $\lim_{\hat{q}\to q}J(x_0,\hat{\mathbf{u}})-J^*(x_0)=0$.
\end{proof}

\begin{remark}
Based on Theorem \ref{performance analysis}, more intuitive and concise upper bounds on the optimality gap can be derived from \eqref{performance gap}. When $\hat{q}<q$, we have $J(x_0,\hat{\mathbf{u}})-J^*(x_0)\le{\rm tr}\{X_{\hat{K}}\}(q-\hat{q})$ because $\hat{P}-P\le0$ (proven by Lemma \ref{assump1}). And when $\hat{q}>q$, there is $J(x_0,\hat{\mathbf{u}})-J^*(x_0)\le {\rm tr}\{X_0\}\lambda_{\max}\{\hat{P}-P\}$. Therefore, the optimality gap can always be directly bounded by the estimated error $q-\hat{q}$ or by the corresponding matrices $\hat{P}-P$. Consider the system given in Example \ref{exa 2} with $x_0\!=\![5;5]$ and assume $q=0.2$. Fig. \ref{gap cost} illustrates the gap $J(x_0,\hat{\mathbf{u}})-J^*(x_0)$ under different $\hat{q}$. As $\hat{q}$ approaches $q$, the gap decreases, whereas as $\hat{q}$ diverges from $q$, the gap increases. If the sample size approaches infinity, the optimality gap tends to zero almost surely. Furthermore, as shown in Fig. \ref{gap cost}, using a larger estimate to design $\hat{K}$ is a conservative approach that may incur higher costs. For example, the optimality gap at $\hat{q}=0.4$ is clearly greater than the cost at $\hat{q}=0$.
\end{remark}

\section{Simulations}\label{section6}
This section presents numerical simulations of two systems to validate the main results.
\begin{example}\rm\label{ex 1}
    Consider a scalar system with $A=1.5$, $B=Q=R=1$. The stabilizing regions (represented in blue) distinguished by  \Cref{Th2} and \Cref{scalar th3} are shown in Fig. \ref{th3 scalar hat_q--q} and \ref{th4 hat_q--q}, respectively. The unstabilizable regions (red regions) are searched with a step size $0.001$ of $q$ and $\hat{q}$ based on \Cref{thm n and s scalar}. The gray areas should ideally be blue, but they cannot be distinguished by the \Cref{Th2} and \Cref{scalar th3}.
    The sample complexity curves given by \eqref{th5_36} and \eqref{scalar sample complexity} are shown in Fig. \ref{th6 scalar sample complexity} and \ref{sample complexity tailored for scalar}, respectively. 
    When $q<0.128$ in Fig. \ref{th6 scalar sample complexity} and $q<0.231$ in Fig. \ref{sample complexity tailored for scalar}, $\hat{K}$ is stabilizing even if $N_q=0$.
     For scalar systems, Theorem \ref{scalar th3} exhibits less conservatism than Theorem \ref{Th2}, as Theorem \ref{scalar th3} is the tailored result. Hence, the sample complexity based on the tailored result \eqref{scalar sample complexity} is also significantly lower than that based on the general result \eqref{th5_36}.

\end{example}
\begin{example}\rm\label{exa 2}
    Given a two-dimensional system with $A=[1.5,0.1;0,1]$, $B\!=\!Q\!=\!R\!=\![1,0;0,1]$.
    The stabilizing regions distinguished by \Cref{Th2} and \Cref{Th 3} are shown in Fig. \ref{th3 vector hat_q--q} and \ref{th5 hat_q--q}, respectively, where the meanings of the different colored regions correspond to those in \Cref{ex 1}. Particularly, the red regions are searched using a step size $0.001$ of $q,\hat{q}$ based on the conclusion from \cite{xiaonan2012} that $\hat{K}$ is stabilizing if and only if $\rho(\Phi)<1$ where $\Phi\triangleq[A+(1-q)B\hat{K}]\otimes[A+(1-q)B\hat{K}]+(q-q^2)(B\hat{K})\otimes(B\hat{K})$. Clearly, Theorem \ref{Th 3} exhibits less conservatism than Theorem \ref{Th2} for this specific system, which makes the sample complexity based on \eqref{invertible B sample complexity} lower than that based on \eqref{th5_36}, as shown in Fig. \ref{th6 vector sample complexity},\ref{sample complexity tailored for vector}. When $q<0.104$ in Fig. \ref{th6 vector sample complexity} and $q<0.167$ in Fig. \ref{sample complexity tailored for vector}, $\hat{K}$ is stabilizing even if $N_q=0$.
    
\end{example}

The error bounds provided by Theorem \ref{Th2},\ref{scalar th3},\ref{Th 3} for \Cref{ex 1} and \Cref{exa 2} under different $q$ are presented in Fig. \ref{error bound}. And the conservatism curves of Theorem \ref{Th2}, \ref{scalar th3},\ref{Th 3} for \Cref{ex 1} and \Cref{exa 2} are depicted in Fig. \ref{conservatism}.
\begin{figure}[h]
    \centering
    \includegraphics[width=0.8\linewidth]{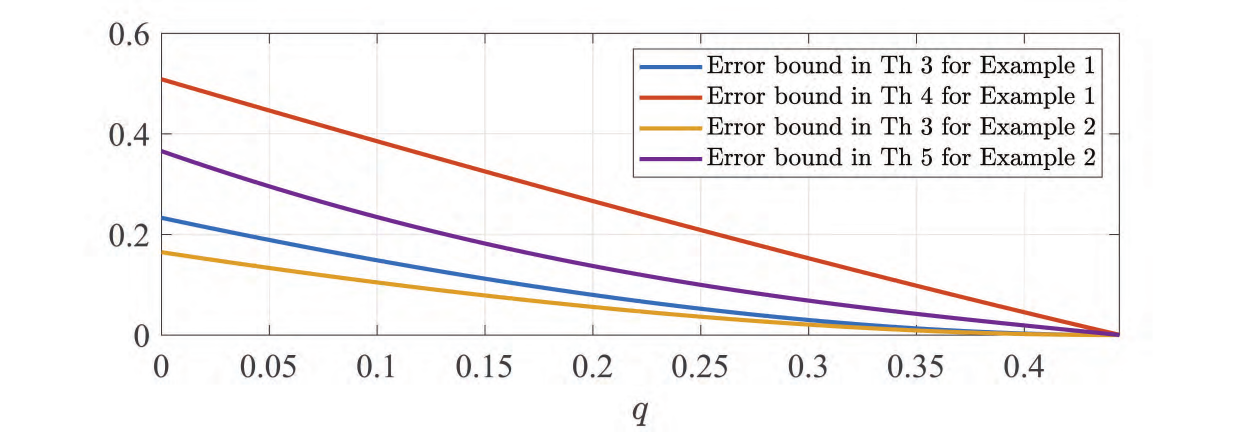}
    \caption{Bounds on the estimation error $q-\hat{q}$ provided in Theorem \ref{Th2},\ref{scalar th3},\ref{Th 3} for Example \ref{ex 1},\ref{exa 2}.}
    \label{error bound}
\end{figure}

\begin{figure}[h]
    \centering
    \includegraphics[width=0.8\linewidth]{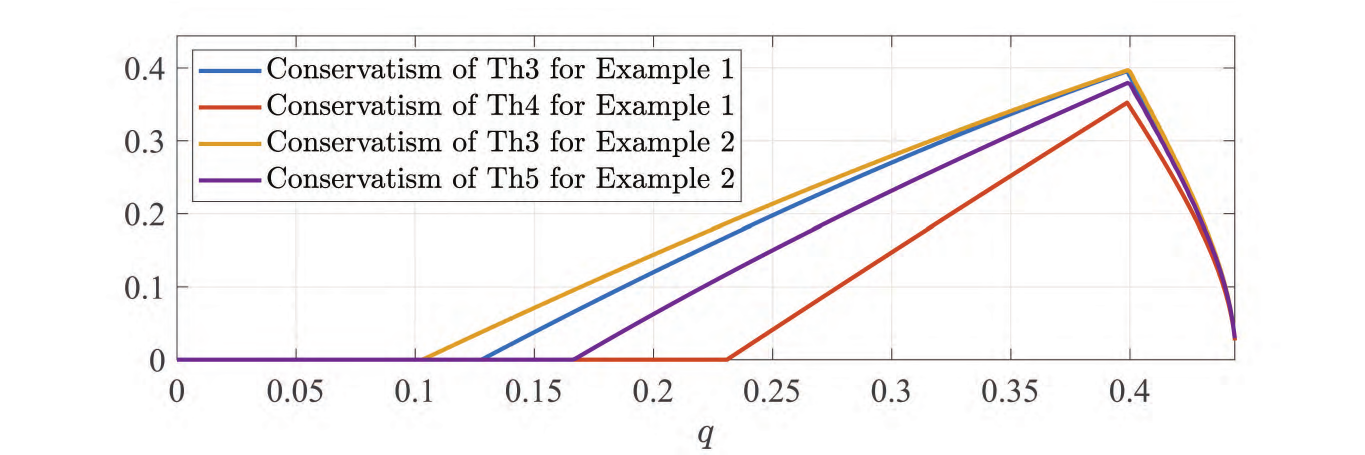}
    \caption{Conservatism of Theorem \ref{Th2},\ref{scalar th3},\ref{Th 3} refers to the extent to which the minimum stabilizable $\hat{q}$ given by these theorems exceeds that given by the necessary and sufficient conditions.}
    \label{conservatism}
\end{figure}

\section{Conclusion}\label{section7}
This paper investigates the LQR problem over an unknown Bernoulli packet loss channel, where the unknown probability is estimated using finite samples. We analyze the maximum estimation error and the minimum sample size for the CE controller to stabilize the system, and the optimality gap between the CE controller and the optimal controller. In the future, this work can be extended to the contexts of LQG and Markov packet loss channels. Furthermore, we aim to design an online controller, where the estimated probability can be improved over time, and analyze its performance.

\appendices

\end{document}